\documentclass[a4paper,fleqn,11pt]{article}

\usepackage{amsmath}
\usepackage{amsthm}
\usepackage{amssymb}
\usepackage{accents}

\usepackage[a4paper,top=3cm, bottom=3cm, left=3cm, right=3cm]{geometry}
\usepackage[shortlabels,inline]{enumitem}
\usepackage[pdftex]{graphicx}
\usepackage{subcaption}
\usepackage[pdftex,ocgcolorlinks,pagebackref=false]{hyperref}
\usepackage[affil-it]{authblk}

\setlist[enumerate,1]{label={(\roman*)}}

\theoremstyle{plain}
\newtheorem{theorem}{Theorem}[section]
\newtheorem{lemma}[theorem]{Lemma}
\newtheorem{proposition}[theorem]{Proposition}
\newtheorem{corollary}[theorem]{Corollary}
\newtheorem{remark}[theorem]{Remark}
\newtheorem{example}[theorem]{Example}

\theoremstyle{definition}

\newtheorem{definition}[theorem]{Definition}

\newcommand{\ed}{\mathop{}\!\mathrm{d}}

\newcommand{\norm}[2][]{\left\|#2\right\|_{#1}}
\newcommand{\ket}[1]{\left|#1\right\rangle}
\newcommand{\bra}[1]{\left\langle #1\right|}
\newcommand{\ketbra}[2]{\left|#1\middle\rangle\!\middle\langle#2\right|}
\newcommand{\braket}[2]{\left\langle#1\middle|#2\right\rangle}
\newcommand{\setbuild}[2]{\left\{#1\middle|#2\right\}}
\DeclareMathOperator{\boundeds}{\mathcal{B}}

\DeclareMathOperator{\substates}{\mathcal{S}_{\le}}
\DeclareMathOperator{\Tr}{Tr}
\DeclareMathOperator{\Hom}{Hom}
\DeclareMathOperator{\support}{supp}
\DeclareMathOperator{\spectrum}{spec}
\DeclareMathOperator{\id}{id}
\newcommand{\relativeentropy}[3][]{\mathop{D_{#1}}\mathopen{}\left(#2\middle\|#3\right)\mathclose{}}
\newcommand{\maxrelativeentropy}[3][]{D^{#1}_{\textnormal{max}}\mathopen{}\left(#2\middle\|#3\right)\mathclose{}}

\newcommand{\complexes}{\mathbb{C}}
\newcommand{\naturals}{\mathbb{N}}
\newcommand{\nonnegativereals}{\mathbb{R}_{\ge 0}}
\newcommand{\positivereals}{\mathbb{R}_{>0}}

\newcommand{\preorderle}{\preccurlyeq}
\newcommand{\preorderge}{\succcurlyeq}
\newcommand{\asymptoticle}{\precsim}
\newcommand{\asymptoticge}{\succsim}
\DeclareMathOperator{\ev}{ev}
\newcommand{\boxes}[1]{\mathcal{B}_{#1}}
\newcommand{\cboxes}[1]{\mathcal{B}_{c,#1}}
\newcommand{\realspectrum}{\Delta}
\newcommand{\testspectrum}{\hat{\Delta}}

\newcommand{\tropicals}{\mathbb{TR}}
\newcommand{\pinching}[1]{\mathcal{P}_{#1}}
\newcommand{\sandwicheddivergence}[3]{\widetilde{D}_{#1}\left(#2\middle\|#3\right)}

\title{Asymptotic relative submajorization of multiple-state boxes}
\author[1,2]{Gergely Bunth}
\author[1,2]{P\'eter Vrana}
\affil[1]{Institute of Mathematics, Budapest University of Technology and Economics, Egry~J\'ozsef u.~1., Budapest, 1111 Hungary.}
\affil[2]{MTA-BME Lend\"ulet Quantum Information Theory Research Group}

\begin{document}
\maketitle
\begin{abstract}
Pairs of states, or ``boxes'' are the basic objects in the resource theory of asymmetric distinguishability (Wang and Wilde, 2019), where free operations are arbitrary quantum channels that are applied to both states. From this point of view, hypothesis testing is seen as a process by which a standard form of distinguishability is distilled. Motivated by the more general problem of quantum state discrimination, we consider boxes of a fixed finite number of states and study an extension of the relative submajorization preorder to such objects. In this relation a tuple of positive operators is greater than another if there is a completely positive trace nonincreasing map under which the image of the first tuple satisfies certain semidefinite constraints relative to the other one. This preorder characterizes error probabilities in the case of testing a composite null hypothesis against a simple alternative hypothesis, as well as certain error probabilities in state discrimination. We present a sufficient condition for the existence of catalytic transformations between boxes, and a characterization of an associated asymptotic preorder, both expressed in terms of sandwiched R\'enyi divergences. This characterization of the asymptotic preorder directly shows that the strong converse exponent for a composite null hypothesis is equal to the maximum of the corresponding exponents for the pairwise simple hypothesis testing tasks.
\end{abstract}

\section{Introduction}

Resource theories provide a unique viewpoint within numerous areas in quantum information theory and physics, such as entanglement theory and quantum thermodynamics \cite{chitambar2019quantum}. Building upon the work of Matsumoto \cite{matsumoto2010reverse}, Wand and Wilde \cite{wang2019resource} carried out a systematic development of the resource-theoretic approach to hypothesis testing in the form of a resource theory of asymmetric distinguishability (see also \cite{buscemi2019information}, where related results are independently obtained with a different perspective). The objects of this resource theory are pairs of quantum states, ordered by joint transformations with general quantum channels. The task of hypothesis testing is then interpreted as distillation of standard pairs, ``bits of asymmetric distinguishability'', the quantum min- and max-divergences \cite{datta2009min} as well as the quantum relative entropy \cite{hiai1991proper,ogawa2005strong} emerge as the distillable distinguishability and distinguishability costs in single-shot and asymptotic settings.

In \cite{wang2019resource} it is also suggested that the resource theoretic study could be extended to more general discrimination tasks, including the discrimination of more than two states, as in the theory of quantum state discrimination \cite{chefles2000quantum,barnett2009quantum,bae2015quantum}. In this work we take a step in this direction, considering boxes consisting of multiple states. With the aim of incorporating probabilities and approximations in the objects compared, we work with tuples of unnormalized states and introduce a generalization of relative submajorization \cite{renes2016relative} to boxes. We say that a box $(\rho_1,\ldots,\rho_m,\sigma)$ relatively submajorizes $(\rho'_1,\ldots,\rho'_m,\sigma')$ if there is a completely positive trace-nonincreasing map $T$ such that $T(\rho_i)\ge\rho'_i$ and $T(\sigma)\le\sigma'$ as positive semidefinite matrices. For normalized boxes this relation is equivalent to the existence of a joint exact transformation of the states (for classical boxes also to matrix majorization in the sense of \cite{dahl1999matrix}), and for unnormalized boxes it encodes error probabilities in hypothesis testing and state discrimination tasks.

Our main results concern asymptotic, catalytic and many-copy relaxations of the relative submajorization relation between such boxes (see Section~\ref{sec:preliminaries} for precise definitions). We find that in these limits the transformations are governed by certain pairwise R\'enyi divergences, and as in \cite{perry2020semiring}, the relevant quantum extensions are the sandwiched R\'enyi divergences \cite{muller2013quantum,wilde2014strong}. More precisely, given (unnormalized) boxes $(\rho_1,\ldots,\rho_m,\sigma)$ and $(\rho'_1,\ldots,\rho'_m,\sigma')$, where $\sigma$ and $\sigma'$ are invertible, we prove the following:
\begin{enumerate}
\item\label{it:mainequivalence} $(\rho_1,\ldots,\rho_m,\sigma)$ asymptotically relative submajorizes $(\rho'_1,\ldots,\rho'_m,\sigma')$ iff for all $\alpha\ge 1$ and all $i\in\{1,\ldots,m\}$ the inequalities
\begin{equation}
\Tr\left(\sigma^{\frac{1-\alpha}{2\alpha}}\rho_i\sigma^{\frac{1-\alpha}{2\alpha}}\right)^\alpha\ge\Tr\left({\sigma'}^{\frac{1-\alpha}{2\alpha}}\rho'_i{\sigma'}^{\frac{1-\alpha}{2\alpha}}\right)^\alpha
\end{equation}
hold.
\item\label{it:mainsufficient} If for every $\alpha\ge 1$ and all $i\in\{1,\ldots,m\}$ the \emph{strict} inequalities
\begin{equation}
\Tr\left(\sigma^{\frac{1-\alpha}{2\alpha}}\rho_i\sigma^{\frac{1-\alpha}{2\alpha}}\right)^\alpha>\Tr\left({\sigma'}^{\frac{1-\alpha}{2\alpha}}\rho'_i{\sigma'}^{\frac{1-\alpha}{2\alpha}}\right)^\alpha
\end{equation}
hold and in addition
\begin{equation}
\norm[\infty]{\sigma^{-1/2}\rho_i\sigma^{-1/2}}>\norm[\infty]{{\sigma'}^{-1/2}\rho'_i{\sigma'}^{-1/2}}
\end{equation}
for all $i\in\{1,\ldots,m\}$, then
\begin{enumerate}[1)]
\item for all sufficiently large $n\in\naturals$ the box $(\rho_1^{\otimes n},\ldots,\rho_m^{\otimes n},\sigma^{\otimes n})$ relatively submajorizes $({\rho'_1}^{\otimes n},\ldots,{\rho'_m}^{\otimes n},{\sigma'}^{\otimes n})$;
\item there exists a catalyst $(\omega_1,\ldots,\omega_m,\tau)$ such that $(\rho_1\otimes\omega_1,\ldots,\rho_m\otimes\omega_m,\sigma\otimes\tau)$ relative submajorizes $(\rho'_1\otimes\omega_1,\ldots,\rho'_m\otimes\omega_m,\sigma'\otimes\tau)$.
\end{enumerate}
\end{enumerate}

As a special case, the normalized box $(\rho_1,\ldots,\rho_m,\sigma)$ can be viewed as a hypothesis testing problem where the states $\rho_1,\ldots,\rho_m$ form a composite null hypothesis, to be tested against the simple alternative hypothesis $\sigma$. In the strong converse regime, a type I error $1-2^{-Rn+o(n)}$ with a type II error $2^{-rn}$ is achievable iff $(\rho_1,\ldots,\rho_m,\sigma)$ asymptotically relative submajorizes the unnormalized box $(2^{-R},\ldots,2^{-R},2^{-r})$. From the characterization \ref{it:mainequivalence} we see that this happens precisely when
\begin{equation}
R\ge\max_i\sup_{\alpha>1}\frac{\alpha-1}{\alpha}\left[r-\sandwicheddivergence{\alpha}{\rho_i}{\sigma}\right].
\end{equation}

We prove these results using recent advances in the theory of preordered semirings \cite{fritz2020local,vrana2020generalization}. In particular, we introduce the semiring of unnormalized boxes, equipped with the preorder given by relative submajorization, and apply two generalizations of Strassen's characterization theorem \cite{strassen1988asymptotic}. The one given in \cite{vrana2020generalization} leads to the equivalence \ref{it:mainequivalence} after classifying the nonnegative real-valued monotone homomorphisms. On the other hand, the implication \ref{it:mainsufficient} is an application of \cite[1.4. Theorem]{fritz2020local}, and in addition requires the classification of the monotone homomorphisms into the tropical semiring. We find that, somewhat surprisingly, both kinds of monotone homomorphisms are pairwise quantities in the sense that each of them depends on only two states of the box.

The proof of our result on asymptotic relative submajorization uses some of the ideas from \cite{perry2020semiring}, where boxes of pairs are studied, but deviates substantially from it in two key steps (in addition to the more obvious differences in the classification of real-valued monotones and the tropical ones that were not considered there). The first difference is that in the case of pairs of states it was possible to find a set of multipliers (the one-dimensional pairs) with the property that every pair can be multiplied by a suitable element in such a way that the product is bounded from above and from below with respect to the natural numbers. This in turn made it possible to use \cite[Theorem 1.2.]{vrana2020generalization}, showing that monotone semiring homomorphisms characterize the asymptotic preorder. With multiple states such a set does not exist, and therefore we must take a different route, effectively applying \cite[Corollary 1.3.]{vrana2020generalization} to the semifield of fractions of the semiring of boxes. The second difference is that in \cite{perry2020semiring} an application of the $\sigma^{\otimes n}$-pinching map was sufficient to ensure that the resulting states commute and thus reduce the evaluation of the monotones on pairs of quantum states to pairs of classical distributions. With more than two states the pinching map alone is not sufficient as the images of different quantum states under the pinching map need not commute. To get around this problem we make use of the special form of the classical monotones.

The remainder of the paper is structured as follows. In Section~\ref{sec:preliminaries} we introduce the relevant notions related to preordered semirings and state recent results relating monotone semiring homomorphisms to several relaxations of the preorder. In Section~\ref{sec:boxsemiring} we introduce the semiring of boxes and extend the relative submajorization preorder to obtain a preordered semiring. In Section~\ref{sec:homomorphisms} we provide a classification of the monotone homomorphisms into the real and tropical real semiring. In Section~\ref{sec:transformations} we derive explicit conditions for asymptotic, many-copy and catalytic relative submajorization in terms of sandwiched R\'enyi divergences. In Section~\ref{sec:statediscrimination} we give applications to asymptotic state discrimination.

\section{Preliminaries}\label{sec:preliminaries}

A \emph{preordered semiring} is a tuple $(S,+,\cdot,0,1,\preorderle)$ where $S$ is a set, $+,\cdot:S\times S\to S$ are commutative and associative binary operations satisfying $(x+y)\cdot z=x\cdot z+y\cdot z$ for all $x,y,z\in S$, $0,1\in S$ are the zero element and the unit (i.e. $0\cdot x=0$ and $1\cdot x=x$ for all $x$), and $\preorderle\subseteq S\times S$ is a transitive and reflexive relation (preorder) such that $x\preorderle y$ implies $x+z\preorderle y+z$ and $x\cdot z\preorderle y\cdot z$ for every $x,y,z\in S$. We will adopt the convention that the binary operations and neutral elements are denoted uniformly with the same symbols $+,\cdot,0,1$ (with the multiplication sign often omitted as usual), and preordered semirings will be referred to via the abbreviated notation $(S,\preorderle)$, indicating only the underlying set and the preorder, or even just $S$ when the preorder is clear. This will in particular be the case in the following examples, where the preorder (in these cases a total order) will be denoted as $\le$.
\begin{example}
The set $\nonnegativereals$ of nonnegative real numbers with its usual addition, multiplication and total order is a preordered semiring.
\end{example}

\begin{example}[tropical semiring in the multiplicative picture]
As a set, the tropical semiring is $\tropicals=\nonnegativereals$, $x+y$ is defined as the maximum of $x$ and $y$, while $\cdot$ is the usual multiplication. We equip this semiring with the usual total order of the real numbers. This is a preordered semiring.
\end{example}

We will be interested in preordered semirings satisfying a pair of additional conditions. First, we require that the canonical map $\naturals\to S$ (the one that sends $n$ to the $n$-term sum $1+1+\cdots 1$) is an order embedding (i.e. injective and $m\le n$ as natural numbers iff their images, also denoted by $m$ and $n$ satisfy $m\preorderle n$). Second, the semiring is assumed to be of \emph{polynomial growth} \cite{fritz2018generalization}. This means that there exists an element $u\in S$ such that $u\preorderge 1$ and for every nonzero $x\in S$ there is a $k\in\naturals$ such that $x\preorderle u^k$ and $1\preorderle u^kx$. Any such element $u$ is called \emph{power universal}. A power universal element need not be unique but the subsequent definitions can be shown not to depend on a particular choice.
\begin{definition}
Let $x,y\in S$. We write $x\asymptoticge y$ and say that $x$ is asymptotically larger than $y$ if for some sublinear sequence $(k_n)_{n\in\naturals}$ of natural numbers and for all $n\in\naturals$ the inequality $u^{k_n}x^n\preorderge y^n$ holds.
\end{definition}

A monotone semiring homomorphism between the preordered semirings $(S_1,\preorderle_1)$ and $(S_2,\preorderle_2)$ is a map $\varphi:S_1\to S_2$ that satisfies $\varphi(0)=0$, $\varphi(1)=1$, $\varphi(x+y)=\varphi(x)+\varphi(y)$, $\varphi(xy)=\varphi(x)\varphi(y)$ and $x\preorderle_1 y\implies\varphi(x)\preorderle_2\varphi(y)$ for $x,y\in S_1$. We will consider monotone homomorphisms into the real and tropical real semirings. For these we introduce the following notations: given a preordered semiring $(S,\preorderle)$ we let $\realspectrum(S,\preorderle)=\Hom(S,\nonnegativereals)$ and $\testspectrum(S,\preorderle)=\realspectrum(S,\preorderle)\cup\setbuild{f\in\Hom(S,\tropicals)}{f(u)=2}$. The two parts will be referred to as the real and the tropical part of the \emph{spectrum} of the semiring. It should be noted that while there is an inherent normalization condition in the definition of a homomorphism into the nonnegative reals, there is no such limitation in tropical real valued homomorphisms since one can always rescale in a multiplicative sense by replacing $f(x)$ with $f^c(x)$ for some $c>0$ (see also \cite[Section 13.]{fritz2020local}). This is the reason for requiring that $f(u)=2$ in our definition (the number $2$ itself is arbitrary, but will be convenient relative to our choice of the power universal element $u$ later).

The evaluation map for an element $s\in S$ is the map $\ev_s:\realspectrum(S,\preorderle)\to\nonnegativereals$ defined as $f\mapsto f(s)$ (one could similarly consider the evaluation map on $\testspectrum(S,\preorderle)$, but it is this restricted form that we will need). It should be noted that both kinds of spectra can be endowed with a topology using the evaluation maps and in general $\testspectrum(S,\preorderle)$ is \emph{not} the disjoint union of its real and tropical part as topological spaces.

Our strategy will be to use the elements of the spectrum to characterize the asymptotic preorder. The main tool will be the following result from \cite{vrana2020generalization}.
\begin{theorem}\label{thm:asymptotic}
Let $(S,\preorderle)$ be a preordered semiring of polynomial growth such that $\naturals\hookrightarrow S$ is an order embedding. The following conditions are equivalent:
\begin{enumerate}
\item for every $x,y\in S\setminus\{0\}$ such that $\frac{\ev_y}{\ev_x}:\realspectrum(S,\preorderle)\to\nonnegativereals$ is bounded there is an $n\in\naturals$ such that $nx\asymptoticge y$
\item for every $x,y\in S$ we have $x\asymptoticge y\iff\forall f\in\realspectrum(S,\preorderle):f(x)\ge f(y)$.
\end{enumerate}
\end{theorem}

The asymptotic preorder is not the only relaxation that can be investigated with methods based on monotone homomorphisms. In the recent work \cite{fritz2020local} Fritz has found sufficient conditions for catalytic and multi-copy transformations in terms of monotone homomorphisms into certain semirings including the real and tropical real numbers. We now state a special case of one of these results, specialized to our more restricted setting.
\begin{theorem}[{\cite[second part of 1.4. Theorem, special case]{fritz2020local}}]\label{thm:localglobal}
Let $S$ be a preordered semiring of polynomial growth with $0\preorderle 1$. Suppose that $x,y\in S\setminus\{0\}$ such that for all $f\in\testspectrum(S,\preorderle)$ the strict inequality $f(x)>f(y)$ holds. Then also the following hold:
\begin{enumerate}
\item there is a $k\in\naturals$ such that $u^kx^n\preorderge u^ky^n$ for every sufficiently large $n$
\item if in addition $x$ is power universal then $x^n\preorderge y^n$ for every sufficiently large $n$
\item there is a nonzero $a\in S$ such that $ax\preorderge ay$.
\end{enumerate}
\end{theorem}
In particular, the last condition means that $x$ may be catalytically transformed into $y$ with catalyst $a$ (in \cite{fritz2020local} a catalyst is given explicitly in terms of the $k$ above). We note that any of the listed conditions implies the non-strict inequalities $f(x)\ge f(y)$ for the monotone homomorphisms.

Despite the apparent similarity between the two results quoted above, there seems to be no simple way of reducing one to the other. In the following sections we will apply both in the context of box transformations and develop the results needed to do so in parallel. In particular, we will classify both the real and the tropical real valued monotones so that the implication in Theorem~\ref{thm:localglobal} can be made explicit, and also verify the condition of Theorem~\ref{thm:asymptotic} so that in the presence of non-strict inequalities between the monotones the characterization of the asymptotic preorder is still available.

\section{The semiring of boxes}\label{sec:boxsemiring}

We consider the number $m\in\naturals$, $m\ge 1$ fixed from now on. A \emph{box} is an $m+1$-tuple $(\rho_1,\ldots,\rho_m,\sigma)$ of positive operators on a finite dimensional Hilbert space $\mathcal{H}$ where $\support\sigma=\mathcal{H}$. We allow $\dim\mathcal{H}=0$ in which case there is a unique such tuple. Let us call the boxes $(\rho_1,\ldots,\rho_m,\sigma)$ and $(\rho_1',\ldots,\rho_m',\sigma')$ \emph{equivalent} when there is a unitary $U:\mathcal{H}\to\mathcal{H}'$ such that $\forall i:U\rho_i U^*=\rho_i'$ and $U\sigma U^*=\sigma'$. A box $(\rho_1,\ldots,\rho_m,\sigma)$ will be called \emph{classical} if all pairs of operators commute, and \emph{normalized} if $\Tr\rho_i=\Tr\sigma=1$ for all $i$. A classical box may be identified with a tuple $(p_1,\ldots,p_m,q)$ of measures on a common finite set $\mathcal{X}$ or with a tuple of diagonal positive operators on $\complexes^{\mathcal{X}}$.

We consider the semiring $\boxes{m}$ of equivalence classes of boxes where addition is induced by the direct sum and multiplication is induced by the tensor product. The zero element is the equivalence class of the unique box on any zero dimensional Hilbert space, and the unit is the equivalence class of the box $(1,1,\ldots,1)$ on the Hilbert space $\complexes$ (here we make the identification $\boundeds(\complexes)=\complexes$). We denote the set of equivalence classes of classical boxes by $\cboxes{m}$. $\cboxes{m}$ is a subsemiring of $\boxes{m}$.

We think of a box as a quantum system prepared via an unknown process (``black box''), with $\rho_1,\ldots,\rho_m,\sigma$ representing the possible states the system might be in. This point of view suggests that boxes should be compared by joint transformations, i.e. a box may be transformed into another box precisely when there is a stochastic map (quantum channel) that takes the $i$th state of the initial box into the $i$th state of the final box. This represents the ability of an experimenter to perform a physical process on the unknown state, resulting in a quantum system with different possible states. Formally, we write $(\rho_1,\ldots,\rho_m,\sigma)\preorderge_1(\rho_1',\ldots,\rho_m',\sigma')$ iff there is a channel $T$ such that $T(\rho_i)=\rho'_i$ for all $i$ and $T(\sigma)=\sigma'$. This defines a preorder on $\boxes{m}$, but it unfortunately does not satisfy the requirements of Theorem~\ref{thm:asymptotic} or Theorem~\ref{thm:localglobal}, in particular $0\not\preorderle_1 1$.
\begin{remark}
The general form of Theorem~\ref{thm:localglobal} from \cite[1.4. Theorem]{fritz2020local} does not require $0\preorderle 1$. However, $\preorderge_1$ would still have $m+1$ homomorphisms (the traces of each component) that stay constant under the transformations, preventing strict inequality for comparable pairs. A similar situation is covered in \cite[14.7. Theorem]{fritz2020local}, taking into account, intuitively, the infinitesimal neighborhood of \emph{one} such norm-like homomorphism. It is quite possible that an analogous result can be obtained that is able to handle multiple conserved values, and we expect this to be an interesting line of research.
\end{remark}

Following a similar route as in \cite{perry2020semiring}, we work instead with a relaxed preorder that ensures $0\preorderle 1$ and that generalizes the relative submajorization preorder defined for pairs of states in \cite{renes2016relative}:
\begin{definition}\label{def:boxpreorder}
Let $(\rho_1,\ldots,\rho_m,\sigma)$ and $(\rho_1',\ldots,\rho_m',\sigma')$ be boxes on $\mathcal{H}$ and $\mathcal{H}'$, respectively. We write $(\rho_1,\ldots,\rho_m,\sigma)\preorderge(\rho_1',\ldots,\rho_m',\sigma')$ iff there exists a completely positive trace non-increasing map $T:\boundeds(\mathcal{H})\to\boundeds(\mathcal{H}')$ such that the following inequalities hold (with the semidefinite partial order):
\begin{equation}\label{eq:boxpreorderdef}
\begin{aligned}
T(\rho_1) & \ge \rho_1'  \\
& \vdots \\
T(\rho_m) & \ge \rho_m'  \\
T(\sigma) & \le \sigma'
\end{aligned}
\end{equation}
\end{definition}
For $m=1$ the relation is identical to relative submajorization \cite[Definition 3]{renes2016relative}.

It should be noted that some of the relations in Definition~\ref{def:boxpreorder} can often be upgraded to equalities by an appropriate modification of the map $T$. Specifically, this is the case if $\Tr\sigma\ge\Tr\sigma'$:
\begin{proposition}\label{prop:boxpreorderstronger}
Let $(\rho_1,\ldots,\rho_m,\sigma)$ and $(\rho_1',\ldots,\rho_m',\sigma')$ be boxes such that the inequality $(\rho_1,\ldots,\rho_m,\sigma)\preorderge(\rho_1',\ldots,\rho_m',\sigma')$ holds. Then there exists a completely positive trace nonincreasing map $\tilde{T}:\boundeds(\mathcal{H})\to\boundeds(\mathcal{H}')$ such that the inequalities similar to \eqref{eq:boxpreorderdef} are satisfied and in addition
\begin{enumerate}
\item if $\Tr\sigma\ge\Tr\sigma'$ then $\tilde{T}(\sigma)=\sigma'$
\item if $\Tr\sigma=\Tr\sigma'$ then $\tilde{T}$ is trace preserving 
\end{enumerate}

Moreover, if $\Tr\rho_i\le\Tr\rho_i'$ for some $i$, then also $T(\rho_i)=\rho_i'$ for any map $T$ satisfying \eqref{eq:boxpreorderdef}.
\end{proposition}
\begin{proof}
Let $T:\boundeds(\mathcal{H})\to\boundeds(\mathcal{H}')$ be a completely positive trace non-increasing map satisfying \eqref{eq:boxpreorderdef}. Define the map $\tilde{T}$ as
\begin{equation}
\tilde{T}(X)=T(X)+[\Tr X-\Tr T(X)]\tau
\end{equation}
for some $\tau\in\substates(\mathcal{H}')$ to be specified later. Then $\tilde{T}$ is a sum of completely positive maps, therefore also completely positive. It is also trace nonincreasing since
\begin{equation}
\Tr\tilde{T}(X)=\Tr T(X)+[\Tr X-\Tr T(X)]\Tr\tau\le\Tr T(X)+[\Tr X-\Tr T(X)]=\Tr X,
\end{equation}
which also shows that $\tilde{T}$ is trace preserving iff $\Tr\tau=1$. The inequalities involving $\rho_i$ are still satisfied because $\tau\ge 0$:
\begin{equation}
\tilde{T}(\rho_i)=T(\rho_i)+[\Tr \rho_i-\Tr T(\rho_i)]\tau\ge T(\rho_i)\ge\rho_i'.
\end{equation}

It remains to choose $\tau$ in such a way that \eqref{eq:boxpreorderdef} is satisfied. If $\Tr\sigma=\Tr T(\sigma)$ then $T(\sigma)\le\sigma'$ and $\Tr\sigma\ge\Tr\sigma'$ implies $T(\sigma)=\sigma'$, in which case any $\tau$ will do. Otherwise $\Tr T(\sigma)<\Tr\sigma$ and we can choose
\begin{equation}
\tau=\frac{\sigma'-T(\sigma)}{\Tr\sigma-\Tr T(\sigma)}.
\end{equation}

This choice ensures
\begin{equation}
\tilde{T}(\sigma)=T(\sigma)+[\Tr\sigma-\Tr T(\sigma)]\frac{\sigma'-T(\sigma)}{\Tr\sigma-\Tr T(\sigma)}=T(\sigma)+\sigma'-T(\sigma)=\sigma',
\end{equation}
and in addition $\Tr\tau=1$ iff $\Tr\sigma=\Tr\sigma'$.

For the last claim we only need to observe that if $T(\rho_i)\ge\rho_i'$ and $\Tr\rho_i\le\Tr\rho_i'$ then in fact $\Tr\rho_i=\Tr T(\rho_i)=\Tr\rho_i'$ and therefore also $T(\rho_i)=\rho_i'$.
\end{proof}

\begin{proposition}
$(\boxes{m})$ is a preordered semiring.
\end{proposition}
\begin{proof}
We need to verify that the preorder is compatible with the semiring operations. Suppose that $(\rho_1,\ldots,\rho_m,\sigma)\preorderge(\rho_1',\ldots,\rho_m',\sigma')$ and let $T$ be a completely positive trace non-increasing map as in Definition~\ref{def:boxpreorder}. Let $(\omega_1,\ldots,\omega_m,\tau)\in\boxes{m}$ a box on $\mathcal{K}$. Then
\begin{equation}
\begin{aligned}
(T\otimes\id_{\boundeds(\mathcal{K})})(\rho_i\otimes\omega_i) & = T(\rho_i)\otimes\omega_i \ge \rho_1'\otimes\omega_i  \\
(T\otimes\id_{\boundeds(\mathcal{K})})(\sigma\otimes\tau) & = T(\sigma)\otimes\tau \le \sigma'\otimes\tau,
\end{aligned}
\end{equation}
therefore $(\rho_1,\ldots,\rho_m,\sigma)(\omega_1,\ldots,\omega_m,\tau)\preorderge(\rho_1',\ldots,\rho_m',\sigma')(\omega_1,\ldots,\omega_m,\tau)$.

The map $\tilde{T}:\boundeds(\mathcal{H}\oplus\mathcal{K})\to\boundeds(\mathcal{H}'\oplus\mathcal{K})$ defined as
\begin{equation}
\tilde{T}\left(\begin{bmatrix}
A & B  \\
C & D
\end{bmatrix}
\right)=\begin{bmatrix}
T(A) & 0  \\
0 & D
\end{bmatrix}
\end{equation}
is also completely positive and trace non-increasing, and satisfies
\begin{equation}
\begin{aligned}
\tilde{T}(\rho_i\oplus\omega_i) & = T(\rho_i)\oplus\omega_i \ge \rho_1'\oplus\omega_i  \\
\tilde{T}(\sigma\oplus\tau) & = T(\sigma)\oplus\tau \le \sigma'\oplus\tau,
\end{aligned}
\end{equation}
therefore $(\rho_1,\ldots,\rho_m,\sigma)+(\omega_1,\ldots,\omega_m,\tau)\preorderge(\rho_1',\ldots,\rho_m',\sigma')+(\omega_1,\ldots,\omega_m,\tau)$.
\end{proof}

\section{Classification of the monotone homomorphisms}\label{sec:homomorphisms}

We turn to the classification of monotone real and tropical real valued monotones. First we consider only classical boxes, relying heavily on the special structure of the semiring $\cboxes{m}$. We will see below that the box $u=(2,2,\ldots,2,1)$ on $\complexes$ is power universal. In this section we do not use this property but we choose this element for the normalization of the tropical real-valued monotontes. A classical box $(p_1,\ldots,p_m,q)$ on $\complexes^{\mathcal{X}}$ is characterized by the diagonal elements $(p_i)_x$ and $q_x$ ($i\in\{1,\ldots,m\}$, $x\in\mathcal{X}$).
\begin{theorem}
$\realspectrum(\cboxes{m},\preorderle)$ consists of the maps
\begin{equation}\label{eq:crealmonotones}
f_{\alpha,i}(p_1,\ldots,p_m,q)=\sum_{x\in\mathcal{X}}(p_i)_x^\alpha q_x^{1-\alpha}
\end{equation}
where $i\in\{1,\ldots,m\}$ and $\alpha\in[1,\infty)$.

The monotone homomorphisms from $\cboxes{m}$ to $\tropicals$ that send $u=(2,2,\ldots,2,1)$ to $2$ are the maps
\begin{equation}\label{eq:ctropicalmonotones}
f_{\infty,i}(p_1,\ldots,p_m,q)=\max_{x\in\mathcal{X}}\frac{(p_i)_x}{q_x}.
\end{equation}
\end{theorem}
\begin{proof}
From the expressions above it is clear that $f_{\alpha,i}$ is a semiring-homomorphism into $\nonnegativereals$ when $\alpha<\infty$ and into $\tropicals$ when $\alpha=\infty$. The maps $f_{\alpha,i}$ are related to the R\'enyi divergences as $f_{\alpha,i}(p_1,\ldots,p_m,q)=2^{(\alpha-1)\relativeentropy[\alpha]{p_i}{q}}$ when $\alpha\in[1,\infty)$ and $f_{\infty,i}(p_1,\ldots,p_m,q)=2^{\relativeentropy[\infty]{p_i}{q}}$, therefore monotone under relations of the form
\begin{equation}
(p_1,\ldots,p_m,q)\preorderge(T(p_1),\ldots,T(p_m),T(q)),
\end{equation}
where $T$ is completely positive and trace preserving (by the data processing inequality). They are also monotone under projections onto subsets of $\mathcal{X}$, since every term in the sum is nonnegative (and since we take the maximum over $\mathcal{X}$). These two operations generate every completely positive trace nonincreasing map. Finally, one verifies that the maps are increasing in $p_i$ and decreasing in $q$.

We show that these are the only elements of the spectrum. Let $f:\cboxes{m}\to\nonnegativereals$ or $f:\cboxes{m}\to\tropicals$ be a monotone homomorphism and consider the functions $g(x)=f(x,\ldots,x,x)$ and $h_i(y)=f(1,\ldots,1,y,1,\ldots,1)$ (with $y$ at the $i$th position), where $x,y\in\positivereals$ and the arguments are one-dimensional boxes. $g$ and $h_i$ inherit the multiplicativity of $f$ and are monotone increasing ($h_i$ essentially by definition, while if $0\le x_1<x_2$ then the map $T=\frac{x_1}{x_2}\id_{\boundeds(\complexes)}$ shows that $(x_1,\ldots,x_1)\preorderle(x_2,\ldots,x_2)$). This implies that $g(x)=x^\beta$ and $h_i(y)=y^{\alpha_i}$ for some $\alpha_i,\beta\ge0$. In addition, the map $x\mapsto f(1,\ldots,1,x)$ is monotone decreasing, therefore $\beta-\sum_{i=1}^m\alpha_i\le 0$.

From this point we reason for the two types of homomorphisms separately, the most obvious difference being that the value of $\beta$ depends on the type. For elements in the real part of the spectrum we derive a convexity condition that is necessary for a homomorphism to be monotone, while for the tropical part it is replaced by quasi-convexity. With hindsight one can see that the formally weaker joint quasiconvexity constraint already excludes every combination of the exponents that is not allowed by joint convexity, but we find it instructive to include both arguments.

Consider first a homomorphism $f$ into $\nonnegativereals$ and boxes $(p_1,\ldots,p_m,q)$ with full support ($\support p_1=\cdots=\support p_m=\support q=\complexes^{\mathcal{X}}$). We have
\begin{equation}
\begin{split}
f(p_1,\ldots,p_m,q)
 & = \sum_{x\in\mathcal{X}}f((p_1)_x,\ldots,(p_m)_x,q_x)  \\
 & = \sum_{x\in\mathcal{X}}g(q_x)\prod_{i=1}^m h_i\left(\frac{(p_i)_x}{q_x}\right)  \\
 & = \sum_{x\in\mathcal{X}}q_x^{\beta-\sum_{i=1}^m\alpha_i}\prod_{i=1}^m(p_i)_x^{\alpha_i}.
\end{split}
\end{equation}
In $\cboxes{m}$ the elements $(2,2,\ldots,2)$ and $1+1=(I_2,\ldots,I_2)$ are equivalent in the sense that $(2,2,\ldots,2)\preorderge(I_2,\ldots,I_2)$ (choose $T(x)=x\frac{I_2}{2}$) and $(I_2,\ldots,I_2)\preorderge(2,2,\ldots,2)$ (choose $T(x)=\Tr x$). Applying monotonicity and additivity we get $g(2)=2g(1)$, and therefore $\beta=1$.

Let $p_1,\ldots,p_m,q,p_1',\ldots,p_m',q'>0$ and $\lambda\in(0,1)$. Choosing $T=\Tr$ we see that
\begin{multline}\label{eq:convexcombination}
\left(\begin{bmatrix}
\lambda p_1 & 0  \\
0 & (1-\lambda)p_1'
\end{bmatrix},\ldots,\begin{bmatrix}
\lambda p_m & 0  \\
0 & (1-\lambda)p_m'
\end{bmatrix},\begin{bmatrix}
\lambda q & 0  \\
0 & (1-\lambda)q'
\end{bmatrix}
\right)  \\  \preorderge(\lambda p_1+(1-\lambda)p_1',\ldots,\lambda p_m+(1-\lambda)p_m',\lambda q+(1-\lambda)q'),
\end{multline}
and therefore
\begin{multline}
\lambda f(p_1,\ldots,p_m,q)+(1-\lambda)f(p_1',\ldots,p_m',q')  \\
 = f(\lambda p_1,\ldots,\lambda p_m,\lambda q)+f((1-\lambda)p_1',\ldots,(1-\lambda)p_m',(1-\lambda)q')  \\
 \ge f(\lambda p_1+(1-\lambda)p_1',\ldots,\lambda p_m+(1-\lambda)p_m',\lambda q+(1-\lambda)q'),
\end{multline}
i.e. $f$ is jointly convex on boxes on $\complexes$ (thought of as a map $\positivereals^{m+1}\to\positivereals$). With the abbreviation $\delta=(1-\sum_i\alpha_i)$ its Hesse matrix at $(p_1,\ldots,p_m,q)=(1,\ldots,1)$ is
\begin{equation}
\begin{bmatrix}
\alpha_1^2-\alpha_1 & \alpha_1\alpha_2 & \cdots & \alpha_1\alpha_m & \alpha_1\delta  \\
\alpha_2\alpha_1 & \alpha_2^2-\alpha_2 & \cdots & \alpha_2\alpha_m & \alpha_2\delta  \\
\vdots & \ddots & \ddots & \ddots & \vdots  \\
\alpha_m\alpha_1 & \alpha_m\alpha_2 & \cdots & \alpha_m^2-\alpha_m & \alpha_m\delta  \\
\delta\alpha_1 & \delta\alpha_2 & \cdots & \delta\alpha_m & \delta^2-\delta
\end{bmatrix}=A-D,
\end{equation}
where $A=\begin{bmatrix}
\alpha_1 & \cdots & \alpha_m & \delta
\end{bmatrix}^T\cdot\begin{bmatrix}
\alpha_1 & \cdots & \alpha_m & \delta
\end{bmatrix}$ and $D$ is a diagonal matrix with entries $\alpha_1,\ldots,\alpha_m,\delta$. The difference must be positive semidefinite. Since $A$ has rank $1$, $D$ can have at most one strictly positive eigenvalue. If there are none, then $\delta=1>0$, a contradiction. Therefore there is a unique index $i$ such that $\alpha_i>0$ and for $i'\neq i$ we have $\alpha_{i'}=0$. From the condition $1-\alpha_i=\delta\le 0$ we get $\alpha_i\ge 1$, i.e. \eqref{eq:crealmonotones} is the only possible form.

Now let $f$ be a homomorphism into $\tropicals$. Then for boxes with full support we have
\begin{equation}
\begin{split}
f(p_1,\ldots,p_m,q)
 & = \max_{x\in\mathcal{X}}f((p_1)_x,\ldots,(p_m)_x,q_x)  \\
 & = \max_{x\in\mathcal{X}}g(q_x)\prod_{i=1}^m h_i\left(\frac{(p_i)_x}{q_x}\right)  \\
 & = \max_{x\in\mathcal{X}}q_x^{\beta-\sum_{i=1}^m\alpha_i}\prod_{i=1}^m(p_i)_x^{\alpha_i}.
\end{split}
\end{equation}
Comparing $(2,2,\ldots,2)$ and $1+1$ again, we now get $g(2)=g(1)$, and therefore $\beta=0$. From the normalization condition $f(u)=2$ we get $\sum_{i=1}^m\alpha_i=1$. Applying $f$ to \eqref{eq:convexcombination} results in the inequality
\begin{multline}
\max\{f(p_1,\ldots,p_m,q), f(p_1',\ldots,p_m',q)\}  \\
 \ge f(\lambda p_1+(1-\lambda)p_1',\ldots,\lambda p_m+(1-\lambda)p_m',\lambda q+(1-\lambda)q'),
\end{multline}
i.e. this time $f$ is jointly quasiconvex on one-dimensional boxes (again, as a map $\positivereals^{m+1}\to\positivereals$). In particular, if we restrict $f$ to a line segment then it is not possible to have a zero directional derivative and negative second derivative (strict local maximum). Suppose that there are two distinct indices $i<j$ such that $\alpha_i,\alpha_j>0$. We consider the point $(1,\ldots,1,1)$ and the direction $(0,\ldots,0,\alpha_j,0,\ldots,0,-\alpha_i,0,\ldots,0)$, where $\alpha_j$ is the $i$th component and $-\alpha_i$ is the $j$th component. The derivaties are
\begin{equation}
\left.\frac{\ed}{\ed s}f(1,\ldots,1,1+s\alpha_j,1,\ldots,1,1-s\alpha_i,1,\ldots,1)\right|_{s=0}=\alpha_i\alpha_j+\alpha_j(-\alpha_i)=0
\end{equation}
\begin{multline}
\left.\frac{\ed^2}{\ed s^2}f(1,\ldots,1,1+s\alpha_j,1,\ldots,1,1-s\alpha_i,1,\ldots,1)\right|_{s=0}  \\
 = \alpha_j^2\alpha_i(\alpha_i-1)+\alpha_i^2\alpha_j(\alpha_j-1)+2\alpha_j(-\alpha_i)\alpha_i\alpha_j
 = -\alpha_i\alpha_j(\alpha_i+\alpha_j)<0,
\end{multline}
a contradiction. Thus there is only one nonzero $\alpha_i$ which, by normalization, has to be $1$.

Finally, the extension to general classical boxes with possibly unequal supports follows from a continuity argument as in \cite{perry2020semiring}: if $(p_1,\ldots,p_m,q)$ is any classical box, then we have
\begin{multline}\label{eq:approximatefullsupport}
(p_1+\epsilon q,\ldots,p_m+\epsilon q,q)\preorderge(p_1,\ldots,p_m,q)  \\  \preorderge\left(\epsilon\frac{\norm[1]{p_1}}{\norm[1]{q}}q+(1-\epsilon)p_1,\ldots,\epsilon\frac{\norm[1]{p_1}}{\norm[1]{q}}q+(1-\epsilon)p_1,q\right)
\end{multline}
for every $\epsilon\in(0,1)$ (in the first inequality choosing $T=\id$, in the second one $T(x)=\epsilon\Tr(x)\frac{q}{\norm[1]{q}}+(1-\epsilon)x$ in the definition). Let $f$ be an element of the $\testspectrum(\boxes{c},\preorderle)$ and apply to \eqref{eq:approximatefullsupport} to get an upper and a lower bound on $f(p_1,\ldots,p_m,q)$. The bounds are of the form found above since all the supports are now equal to $\support q$ and we can see that they converge to the same value as $\epsilon\to 0$.
\end{proof}

We turn to the semiring of quantum boxes, using that any monotone homomorphism on $\boxes{m}$ must restrict to a monotone homomorphism on $\cboxes{m}$. As in \cite{perry2020semiring}, we show that there is only one possible extension. In that case it was possible to reduce the evaluation at a quantum pair to a classical one using the pinching map.  However, that argument needs to be modified because it is not possible to transform a quantum box of multiple states into a classical one using a pinching map alone, since different pinched states need not commute with each other. 

The pinching map is defined as follows (see \cite{hayashi2002optimal} or \cite[Section 2.6.3]{tomamichel2015quantum}). Let $\mathcal{H}$ be a finite dimensional Hilbert space and $A\in\boundeds(\mathcal{H})$ a normal operator with spectral decomposition
\begin{equation}
A=\sum_{\lambda\in\spectrum(A)}\lambda P_\lambda,
\end{equation}
where $(P_\lambda)_{\lambda\in\complexes}$ are pairwise disjoint orthogonal projections summing to $I$. The pinching map $\pinching{A}:\boundeds(\mathcal{H})\to\boundeds(\mathcal{H})$ is defined as
\begin{equation}
\pinching{A}(X)=\sum_{\lambda\in\spectrum(A)}P_{\lambda}XP_{\lambda}.
\end{equation}
Any operator in the image commutes with $A$ and for $X\ge0$ it satisfies the pinching inequality $|\spectrum(A)|\pinching{A}(X)\ge X$.

\begin{theorem}\label{thm:qmonotones}
$\realspectrum(\boxes{m},\preorderle)$ consists of the maps
\begin{equation}
\tilde{f}_{\alpha,i}(\rho_1,\ldots,\rho_m,\sigma)=\Tr\left(\sigma^{\frac{1-\alpha}{2\alpha}}\rho_i\sigma^{\frac{1-\alpha}{2\alpha}}\right)^\alpha
\end{equation}
where $i=1,2,\ldots,m$ and $\alpha\in[1,\infty)$.

The monotone homomorphisms from $\boxes{m}$ to $\tropicals$ that send $u=(2,2,\ldots,2,1)$ to $2$ are the maps
\begin{equation}
\tilde{f}_{\infty,i}(\rho_1,\ldots,\rho_m,\sigma)=\norm[\infty]{\sigma^{-1/2}\rho_i\sigma^{-1/2}}.
\end{equation}
\end{theorem}
\begin{proof}
It is readily verified that the expressions above give homomorphisms into the respective semirings, monotone by the data processing inequality for the sandwiched (or minimal) R\'enyi divergences \cite{mosonyi2015quantum}. We will show that these are the only possible extensions of the elements of $\testspectrum(\cboxes{m},\preorderle)$ to $\boxes{m}$.

Let $\alpha\in[1,\infty)$ and $i\in\{1,\ldots,m\}$, and let $\tilde{f}$ be any extension of $f_{\alpha,i}$. For a large enough $c\in\positivereals$ and every $n\in\naturals$ we have from the pinching inequality
\begin{equation}\label{eq:qextensionupper}
(c^nI^{\otimes n},\ldots,c^nI^{\otimes n},|\spectrum(\sigma^{\otimes n})|\pinching{\sigma^{\otimes n}}(\rho_i^{\otimes n}),c^nI^{\otimes n},\ldots,c^nI^{\otimes n},\sigma^{\otimes n})\preorderge(\rho_1^{\otimes n},\ldots,\rho_m^{\otimes n},\sigma^{\otimes n}),
\end{equation}
where the left hand side is classical, therefore
\begin{equation}\label{eq:pinchingupper}
\begin{split}
\tilde{f}(\rho_1,\ldots,\rho_m,\sigma)
 & \le \sqrt[n]{f_{\alpha,i}(c^nI^{\otimes n},\ldots,c^nI^{\otimes n},|\spectrum(\sigma^{\otimes n})|\pinching{\sigma^{\otimes n}}(\rho_i^{\otimes n}),c^nI^{\otimes n},\ldots,c^nI^{\otimes n},\sigma^{\otimes n}}  \\
 & = \sqrt[n]{|\spectrum(\sigma^{\otimes n})|^\alpha\Tr\pinching{\sigma^{\otimes n}}(\rho_i^{\otimes n})^\alpha(\sigma^{\otimes n})^{1-\alpha}}
\end{split}
\end{equation}
For a matching lower bound we consider the inequality ($\ket{0}$ is an additional orthogonal direction)
\begin{multline}\label{eq:qextensionlower}
(\ketbra{0}{0}\oplus\rho_1^{\otimes n},\ldots,\ketbra{0}{0}\oplus\rho_m^{\otimes n},C\ketbra{0}{0}\oplus\sigma^{\otimes n})  \\  \preorderge(\ketbra{0}{0},\ldots,\ketbra{0}{0},\ketbra{0}{0}\oplus\pinching{\sigma^{\otimes n}}(\rho_i^{\otimes n}),\ketbra{0}{0},\ldots,\ketbra{0}{0},C\ketbra{0}{0}\oplus\sigma^{\otimes n}),
\end{multline}
which follows from the definition with $T=\id_{\complexes\ketbra{0}{0}}\oplus\pinching{\sigma^{\otimes n}}$. The right hand side is classical, therefore
\begin{equation}
\begin{split}
C^{1-\alpha}+\tilde{f}(\rho_1,\ldots,\rho_m,\sigma)^n
 & \ge \Tr(\ketbra{0}{0}\oplus\pinching{\sigma^{\otimes n}}(\rho_i^{\otimes n}))^\alpha(\ketbra{0}{0}\oplus\sigma^{\otimes n})^{1-\alpha}  \\
 & = C^{1-\alpha}+\Tr\pinching{\sigma^{\otimes n}}(\rho_i^{\otimes n})^\alpha(\sigma^{\otimes n})^{1-\alpha},
\end{split}
\end{equation}
which leads to
\begin{equation}\label{eq:pinchinglower}
\tilde{f}(\rho_1,\ldots,\rho_m,\sigma) \ge \sqrt[n]{\Tr\pinching{\sigma^{\otimes n}}(\rho_i^{\otimes n})^\alpha(\sigma^{\otimes n})^{1-\alpha}}.
\end{equation}
From \eqref{eq:pinchingupper} and \eqref{eq:pinchinglower} and using $|\spectrum(\sigma^{\otimes n})|^{\alpha/n}\to 1$ we get
\begin{equation}
\tilde{f}(\rho_1,\ldots,\rho_m,\sigma) \ge \lim_{n\to\infty}\sqrt[n]{\Tr\pinching{\sigma^{\otimes n}}(\rho_i^{\otimes n})^\alpha(\sigma^{\otimes n})^{1-\alpha}}=\Tr\left(\sigma^{\frac{1-\alpha}{2\alpha}}\rho\sigma^{\frac{1-\alpha}{2\alpha}}\right)^\alpha,
\end{equation}
where the last equality follows from \cite[Proposition 4.12.]{tomamichel2015quantum} (see also \cite[Theorem 4.4.]{perry2020semiring}).

Consider now an extension $\tilde{f}:\boxes{m}\to\tropicals$ of $f_{\infty,i}$. From \eqref{eq:qextensionupper} we get
\begin{equation}
\tilde{f}(\rho_1,\ldots,\rho_m,\sigma)\le\sqrt[n]{\norm[\infty]{(\sigma^{\otimes n})^{-1/2}|\spectrum(\sigma^{\otimes n})|\pinching{\sigma^{\otimes n}}(\rho_i^{\otimes n})(\sigma^{\otimes n})^{-1/2}}}
\end{equation}
and from \eqref{eq:qextensionlower} we get
\begin{equation}
\max\{C,\tilde{f}(\rho_1,\ldots,\rho_m,\sigma)^n\}\ge\max\{C,\norm[\infty]{(\sigma^{\otimes n})^{-1/2}\pinching{\sigma^{\otimes n}}(\rho_i^{\otimes n})(\sigma^{\otimes n})^{-1/2}}\}.
\end{equation}
For small enough $C$ the maximum equals the second argument on both sides, therefore
\begin{equation}
\tilde{f}(\rho_1,\ldots,\rho_m,\sigma)\ge\sqrt[n]{\norm[\infty]{(\sigma^{\otimes n})^{-1/2}\pinching{\sigma^{\otimes n}}(\rho_i^{\otimes n})(\sigma^{\otimes n})^{-1/2}}}.
\end{equation}
The upper and lower bounds converge as $n\to\infty$ and single out the unique quantum max-divergence (see \cite{datta2009min} and \cite[Section 4.2.4]{tomamichel2015quantum}), i.e.
\begin{equation}
\tilde{f}(\rho_1,\ldots,\rho_m,\sigma)=2^{\maxrelativeentropy{\rho_i}{\sigma}}=\norm[\infty]{\sigma^{-1/2}\rho_i\sigma^{-1/2}}.
\end{equation}
\end{proof}

\section{Conditions for catalytic, multi-copy, and asymptotic relative submajorization}\label{sec:transformations}

We now specialize Theorem~\ref{thm:localglobal} to the preordered semiring $\boxes{m}$. First we verify the polynomial growth condition by exhibiting a power universal element.
\begin{proposition}
The box $u=(2,2,\ldots,2,1)$ on $\complexes$ is power universal.
\end{proposition}
\begin{proof}
Choosing $T=\id_{\boundeds(\complexes)}$ in Definition~\ref{def:boxpreorder} we verify that $u\preorderge 1$.

Let $(\rho_1,\ldots,\rho_m,\sigma)$ be a box on $\mathcal{H}$. We first find $k_1\in\naturals$ such that $u^{k_1}\preorderge(\rho_1,\ldots,\rho_m,\sigma)$. Let $T_1:\boundeds(\complexes)\to\boundeds(\mathcal{H})$ be the map $T_1(x)=2^{-k_1/2}x\sigma$. By choosing $k_1$ large enough we can ensure that $T_1$ is a completely positive trace nonincreasing map. It satisfies $T_1(1)=2^{-k_1/2}\sigma\le\sigma$ and $T_1(2^{k_1})=2^{k_1/2}\sigma$, which is greater than $\rho_i$ when $k_1$ is large, since $\support\rho_i\subseteq\support\sigma$.

To find a $k_2$ such that $u^{k_2}(\rho_1,\ldots,\rho_m,\sigma)\preorderge 1$ we consider the map $T_2:\boundeds(\mathcal{H})\to\boundeds(\complexes)$ given by $T_2(x)=2^{-k_2/2}\frac{\Tr x}{\Tr\sigma}$. This is completely positive and also trace nonincreasing provided that $k_2$ is large enough. By construction, $T_2(\sigma)=2^{-k_2/2}\le 1$. The remaining inequalities
\begin{equation}
1\le T_2(2^{k_2}\rho_i)=2^{k_2/2}\frac{\Tr\rho_i}{\Tr\sigma}
\end{equation}
can also be ensured by choosing $k_2$ large enough.

With $k=\max\{k_1,k_2\}$ we have both $u^k\preorderge(\rho_1,\ldots,\rho_m,\sigma)$ and $u^k(\rho_1,\ldots,\rho_m,\sigma)\preorderge 1$
\end{proof}
Note that in addition to being a power universal element, the box $u$ is also invertible: its multiplicative inverse is the box $u^{-1}=(\frac{1}{2},\ldots,\frac{1}{2},1)$. A consequence is that the implications of Theorem~\ref{thm:localglobal} can be simplified in that one may choose $k=0$ or equivalently, there is no need to assume that $x$ is power universal. Together with Theorem~\ref{thm:qmonotones} this leads to the following condition.
\begin{corollary}
Let $(\rho_1,\ldots,\rho_m,\sigma)$ and $(\rho'_1,\ldots,\rho'_m,\sigma')$ be elements of $\boxes{m}$. Suppose that for every $i\in\{1,\ldots,m\}$ and $\alpha\in[1,\infty)$ the inequalities
\begin{equation}
\Tr\left(\sigma^{\frac{1-\alpha}{2\alpha}}\rho_i\sigma^{\frac{1-\alpha}{2\alpha}}\right)^\alpha>\Tr\left({\sigma'}^{\frac{1-\alpha}{2\alpha}}\rho'_i{\sigma'}^{\frac{1-\alpha}{2\alpha}}\right)^\alpha
\end{equation}
as well as
\begin{equation}
\norm[\infty]{\sigma^{-1/2}\rho_i\sigma^{-1/2}}>\norm[\infty]{{\sigma'}^{-1/2}\rho'_i{\sigma'}^{-1/2}}
\end{equation}
hold. Then
\begin{enumerate}
\item for every sufficiently large $n$ there is a completely positive trace nonincreasing map $T$ such that
\begin{equation}
\begin{aligned}
\forall i\in\{1,\ldots,m\}:T(\rho_i^{\otimes n}) & \ge {\rho'}^{\otimes n}  \\
T(\sigma^{\otimes n}) & \le {\sigma'}^{\otimes n},
\end{aligned}
\end{equation}
\item there is a box $(\tau_1,\ldots,\tau_m,\omega)$ and a completely positive trace nonincreasing map $T$ such that
\begin{equation}
\begin{aligned}
\forall i\in\{1,\ldots,m\}:T(\rho_i\otimes\tau_i) & \ge \rho'_i\otimes\tau_i  \\
T(\sigma_i\otimes\omega_i) & \ge \sigma'_i\otimes\omega_i.
\end{aligned}
\end{equation}
\end{enumerate}
\end{corollary}

Our next goal is to apply Theorem~\ref{thm:asymptotic}. We start with a simple lemma.
\begin{lemma}\label{lem:CPUmap}
Let $A\in\boundeds(\mathcal{H})$, $A'\in\boundeds(\mathcal{H}')$, $A,A'\ge 0$ and $\norm[\infty]{A'}\le\norm[\infty]{A}$. Then there exists a completely positive unital map $\Phi:\boundeds(\mathcal{H})\to\boundeds(\mathcal{H}')$ such that $\Phi(A)\ge\Phi(A')$.
\end{lemma}
\begin{proof}
Let $\psi$ be an eigenvector of $A$ with eigenvalue $\norm[\infty]{A}$ with $\norm{\psi}=1$. Let $\Phi(X)=\bra{\psi}X\ket{\psi}I$. This map is completely positive and unital and satisfies
\begin{equation}
\Phi(A)=\bra{\psi}A\ket{\psi}I=\norm[\infty]{A}I\ge\norm[\infty]{A'}I\ge A'.
\end{equation}
\end{proof}

\begin{proposition}
Let $(\rho_1,\ldots,\rho_m,\sigma)$ and $(\rho_1',\ldots,\rho_m',\sigma)$ be boxes and suppose that
\begin{equation}
\frac{\tilde{f}_{i,\alpha}(\rho'_1,\ldots,\rho'_m,\sigma')}{\tilde{f}_{i,\alpha}(\rho_1,\ldots,\rho_m,\sigma)}
\end{equation}
is bounded for every $i$ as $\alpha\to\infty$. Then there exists an $r\in\mathbb{N}$ such that $(\rho'_1,\ldots,\rho'_m,\sigma')\asymptoticle r\cdot(\rho_1,\ldots,\rho_m,\sigma)$.
\end{proposition}
\begin{proof}
Under the assumptions of the proposition,
\begin{equation}
\begin{split}
\infty
 & > \lim_{\alpha\to\infty}\log\frac{\tilde{f}_{i,\alpha}(\rho'_1,\ldots,\rho'_m,\sigma')}{\tilde{f}_{i,\alpha}(\rho_1,\ldots,\rho_m,\sigma)}  \\
 & = \lim_{\alpha\to\infty}(\alpha-1)\left(\sandwicheddivergence{\alpha}{\rho'_i}{\sigma'}-\sandwicheddivergence{\alpha}{\rho_i}{\sigma}\right).
\end{split}
\end{equation}
The limit of the first factor is $\infty$, therefore the limit of the second factor (which is known to exist) must be at most $0$, i.e. $\sandwicheddivergence{\infty}{\rho'_i}{\sigma'}\le\sandwicheddivergence{\infty}{\rho_i}{\sigma}$. Using the explicit form of the R\'enyi divergence of order $\infty$ we conclude
\begin{equation}
\norm[\infty]{{\sigma'}^{-1/2}\rho'_i{\sigma'}^{-1/2}}\le\norm[\infty]{{\sigma}^{-1/2}\rho_i{\sigma}^{-1/2}}.
\end{equation}
Let $\tilde{\Phi}_i:\boundeds(\mathcal{H})\to\boundeds(\mathcal{H}')$ be a completely positive unital map such that $\tilde{\Phi}_i({\sigma}^{-1/2}\rho_i{\sigma}^{-1/2})\ge{\sigma'}^{-1/2}\rho'_i{\sigma'}^{-1/2}$ (from Lemma~\ref{lem:CPUmap}) and consider the maps
\begin{equation}
\Phi_i(X)={\sigma'}^{1/2}\tilde{\Phi}_i(\sigma^{-1/2}X\sigma^{-1/2}){\sigma'}^{1/2}.
\end{equation}
These are completely positive and satisfy $\Phi_i(\sigma)=\sigma'$ and $\Phi_i(\rho_i)\ge\rho'_i$. Choose $r\in\naturals$ such that
\begin{equation}
r\ge\max_i\norm[1-1]{\Phi_i}
\end{equation}
and let
\begin{equation}
T_n=\frac{1}{mr^n}\sum_{i=1}^m\Phi_i^{\otimes n}.
\end{equation}
For every $n$ this map is trace nonincreasing and therefore
\begin{equation}
\begin{split}
u^{\lfloor\log m\rfloor}\left(r\cdot(\rho_1,\ldots,\rho_m,\sigma)\right)^n
 & \preorderge (mr^n\rho_1^{\otimes n},\ldots,mr^n\rho_m^{\otimes n},r^n\sigma^{\otimes n})  \\
 & \preorderge (T_n(mr^n\rho_1^{\otimes n}),\ldots,T_n(mr^n\rho_m^{\otimes n}),T_n(r^n\sigma^{\otimes n}))  \\
 & \preorderge (\Phi_1^{\otimes n}(\rho_1^{\otimes n}),\ldots,\Phi_m^{\otimes n}(\rho_m^{\otimes n}),T_n(r^n\sigma^{\otimes n}))  \\
 & \preorderge (\rho'_1,\ldots,\rho'_m,\sigma').
\end{split}
\end{equation}
This proves that $r\cdot(\rho_1,\ldots,\rho_m,\sigma)\asymptoticge(\rho'_1,\ldots,\rho'_m,\sigma')$.
\end{proof}

Together with Theorem~\ref{thm:asymptotic} the last proposition and the classification in Theorem~\ref{thm:qmonotones} leads to the following explicit condition.
\begin{corollary}\label{cor:asymptotic}
Let $(\rho_1,\ldots,\rho_m,\sigma)$ and $(\rho'_1,\ldots,\rho'_m,\sigma')$ be elements of $\boxes{m}$. The following are equivalent:
\begin{enumerate}
\item for every $i\in\{1,\ldots,m\}$ and $\alpha\in[1,\infty)$ the inequalities
\begin{equation}
\Tr\left(\sigma^{\frac{1-\alpha}{2\alpha}}\rho_i\sigma^{\frac{1-\alpha}{2\alpha}}\right)^\alpha\ge\Tr\left({\sigma'}^{\frac{1-\alpha}{2\alpha}}\rho'_i{\sigma'}^{\frac{1-\alpha}{2\alpha}}\right)^\alpha
\end{equation}
hold,
\item $(\rho_1,\ldots,\rho_m,\sigma)\asymptoticge(\rho'_1,\ldots,\rho'_m,\sigma')$.
\end{enumerate}
\end{corollary}

\section{Application to state discrimination}\label{sec:statediscrimination}

\subsection{Composite null hypothesis}

One interpretation of a (normalized) box is that the states $\rho_1,\ldots,\rho_m$ form a composite null hypothesis which is to be tested again the simple alternative hypothesis $\sigma$. In this hypothesis testing problem one considers a two-outcome POVM $(\Pi,I-\Pi)$, or \emph{test}, and the decision is based on the measurement result, rejecting the null hypothesis if the second outcome is observed. Such a test is uniquely specified by an operator $\Pi$ such that $0\le T\le I$ and every such operator gives rise to a valid POVM.

A type I error occurs when the null hypothesis is falsely rejected. For every member in the family $\rho_1,\ldots,\rho_m$ we define a probability of type I error,
\begin{equation}
\alpha_i(\Pi)=\Tr\rho_i(I-\Pi)=1-\Tr\rho_i\Pi,
\end{equation}
and the maximum
\begin{equation}
\alpha(\Pi)=\max_i\alpha_i(\Pi)
\end{equation}
is the \emph{significance level} of the test.

In contrast, a type II error means that the correct state was $\sigma$ but the null hypothesis does not get rejected. The probability of a type II error is
\begin{equation}
\beta(\Pi)=\Tr\sigma \Pi.
\end{equation}
In general it is not possible to have a low probability for both types of errors but there is a trade-off between the two quantities. The possible values are exactly characterized by the preordered semiring $(\boxes,\preorderle)$ as follows.
\begin{proposition}\label{prop:compositeerror}
Let $(\rho_1,\ldots,\rho_m,\sigma)$ be a normalized box and $\alpha,\beta\in[0,1]$. The following are equivalent:
\begin{enumerate}
\item there exists a test $\Pi$ with $\alpha(\Pi)\le\alpha$ and $\beta(\Pi)\le\beta$
\item $(\rho_1,\ldots,\rho_m,\sigma)\preorderge((1-\alpha),\ldots,(1-\alpha),\beta)$
\end{enumerate}
\end{proposition}
\begin{proof}
Let $(\rho_1,\ldots,\rho_m,\sigma)$ be a normalized box on $\mathcal{H}$ and suppose that a test exists with the properties above. Consider the map $T:\boundeds(\mathcal{H})\to\boundeds(\complexes)$ given by $T(X)=\Tr(X\Pi)$. $T$  is completely positive because $\Pi\ge 0$ and trace nonincreasing because $\Pi\le I$. We apply $T$ to the box:
\begin{equation}
\begin{aligned}
T(\rho_i) & = \Tr(\rho_i\Pi)=1-\alpha_i(\Pi)\ge 1-\alpha(\Pi)\ge 1-\alpha  \\
T(\sigma) & = \Tr(\sigma\Pi)=\beta(\Pi)\le\beta,
\end{aligned}
\end{equation}
therefore $(\rho_1,\ldots,\rho_m,\sigma)\preorderge((1-\alpha),\ldots,(1-\alpha),\beta)$.

Conversely, suppose that $(\rho_1,\ldots,\rho_m,\sigma)\preorderge((1-\alpha),\ldots,(1-\alpha),\beta)$. This means that there exists a completely positive trace nonincreasing map $T:\boundeds(\mathcal{H})\to\boundeds(\complexes)$ such that $T(\rho_i)\ge 1-\alpha$ and $T(\sigma)\le\beta$. Pick such a map and let $\Pi=T^*(1)$ (where $1=\id_{\complexes}$ is the identity map of $\complexes$). Then $0\le\Pi\le I$ and
\begin{equation}
\begin{aligned}
\alpha_i(\Pi) & = \Tr\rho_i(I-\Pi)=1-\Tr\rho_i T^*(1)=1-T(\rho_i)\le \alpha  \\
\beta(\Pi) & = \Tr\sigma\Pi=\Tr\sigma T^*(1)=T(\sigma)\le\beta,
\end{aligned}
\end{equation}
therefore also $\alpha(\Pi)\le\alpha$.
\end{proof}

Suppose that we have access to $n$ copies of such identically prepared boxes. The resource object describing this situation is the power $(\rho_1,\ldots,\rho_m,\sigma)^n=(\rho_1^{\otimes n},\ldots,\rho_m^{\otimes n},\sigma^{\otimes n})$. If we are allowed to perform a joint measurement then we expect to be able to achieve lower probabilities of both types of errors than with a single copy. In particular, an extension of the quantum Stein lemma says that when $n\to\infty$ and the probability of the type I error is required to go to $0$, it is possible to achieve an exponential decay of the type II error, where the exponent is given by the minimum of the relative entropies $\relativeentropy{\rho_i}{\sigma}$ \cite{bjelakovic2005quantum}.

The asymptotic preorder $\asymptoticge$ is able to capture the exponential decay of the type II error and the exponential convergence of the type I error to \emph{one}, called the strong converse regime. More precisely, we have the following characterization.
\begin{proposition}\label{prop:compositeasymptotic}
The following are equivalent
\begin{enumerate}
\item there is a sequence of tests $\Pi_n$ on $\mathcal{H}^{\otimes n}$ for which the type I error is less than $1-2^{-Rn+o(n)}$ and at the same time the type II error decreases as fast as $2^{-rn}$
\item $(\rho_1,\ldots,\rho_m,\sigma)\asymptoticge(2^{-R},\ldots,2^{-R},2^{-r})$.
\end{enumerate}
\end{proposition}
\begin{proof}
Since $u$ is invertible, the condition appearing in the definition of the asymptotic preorder may be written as
\begin{equation}
\begin{split}
(\rho_1^{\otimes n},\ldots,\rho_m^{\otimes n},\sigma^{\otimes n})
 & \preorderge u^{-k_n}(2^{-R},\ldots,2^{-R},2^{-r})^n  \\
 & = (2^{-Rn-k_n},\ldots,2^{-Rn-k_n},2^{-rn})
\end{split}
\end{equation}
where $k_n/n\to 0$. According to Proposition~\ref{prop:compositeerror}, this is equivalent to the existence of a sequence of tests $\Pi_n$ such that $\alpha(\Pi_n)\le 1-2^{-Rn-k_n}$ and $\beta(\Pi_n)\le 2^{-rn}$ with $k_n\in o(n)$ as claimed.
\end{proof}
To achieve asymptotically the smallest type II error probability for a given exponent $r$, we need to find the smallest $R$ satisfying the equivalent conditions. Denoting this value by $R^*(r)$, Proposition~\ref{prop:compositeasymptotic} and Corollary~\ref{cor:asymptotic} implies
\begin{equation}
R^*(r)=\max_i\sup_{\alpha>1}\frac{\alpha-1}{\alpha}\left[r-\sandwicheddivergence{\alpha}{\rho_i}{\sigma}\right].
\end{equation}
In particular, the exponent is given by the minimum of the pairwise exponents \cite{mosonyi2015quantum} similarly as in the extended Stein lemma.

\subsection{Multiple hypotheses}

In a multiple state discrimination problem one performs a measurement with multiple outcomes, one corresponding to each of the possible states. Mathematically, such a measurement is described by a POVM $(\Pi_1,\ldots,\Pi_m,I-(\Pi_1+\cdots+\Pi_m))$ on a set of size $m+1$. Upon observing the outcome $i$ ($m+1$), the experimenter concludes that the unknown state was $\rho_i$ ($\sigma$). In such a setting one can define $(m+1)m$ different error probabilities depending on which state is incorrectly identified as which other state. Alternatively, one may form $m+1$ probabilities of successful detections (these are of course functionally related to the error probabilities). In our framework it is possible to control $2m$ of these probabilities as follows.
\begin{proposition}\label{prop:multipleerror}
Let $(\rho_1,\ldots,\rho_m,\sigma)$ be a normalized box and $a_1,\ldots,a_m,b_1,\ldots,b_m\in[0,1]$. Let $\ket{1},\ldots,\ket{m}$ be an orthonormal basis in $\complexes^m$ and consider the operator
\begin{equation}
b=\sum_{i=1}^m b_i\ketbra{i}{i}.
\end{equation}
Then the following are equivalent:
\begin{enumerate}
\item there exists a POVM $(\Pi_1,\ldots,\Pi_m,I-(\Pi_1+\cdots+\Pi_m))$ with
\begin{align}
\Tr\rho_i\Pi_i\ge a_i
\intertext{and}
\Tr\sigma\Pi_i\le b_i
\end{align}
for every $1\le i\le m$.
\item $(\rho_1,\ldots,\rho_m,\sigma)\preorderge(a_1\ketbra{1}{1},\ldots,a_m\ketbra{m}{m},b)$.
\end{enumerate}
\end{proposition}
\begin{proof}
Suppose that $(\Pi_1,\ldots,\Pi_m,I-(\Pi_1+\cdots+\Pi_m))$ is a POVM satisfying the conditions. Consider the channel $T:\boundeds(\mathcal{H})\to\boundeds(\complexes^m)$
\begin{equation}
T(X)=\sum_{j=1}^m(\Tr X\Pi_j)\ketbra{j}{j}.
\end{equation}
It satisfies
\begin{align*}
T(\rho_i) & = \sum_{j=1}^m(\Tr \rho_i\Pi_j)\ketbra{j}{j} \ge (\Tr \rho_i\Pi_i)\ketbra{i}{i} \ge a_i\ketbra{i}{i}  \\
T(\sigma) & = \sum_{j=1}^m(\Tr \sigma\Pi_j)\ketbra{j}{j} \le \sum_{j=1}^m b_j\ketbra{j}{j}=b,
\end{align*}
therefore $(\rho_1,\ldots,\rho_m,\sigma)\preorderge(a_1\ketbra{1}{1},\ldots,a_m\ketbra{m}{m},b)$.

Conversely, suppose that $(\rho_1,\ldots,\rho_m,\sigma)\preorderge(a_1\ketbra{1}{1},\ldots,a_m\ketbra{m}{m},b)$. This means that there is a channel $T:\boundeds(\mathcal{H})\to\boundeds(\complexes^m)$ satisfying
\begin{align*}
T(\rho_i) & \ge a_i\ketbra{i}{i}  \\
T(\sigma) & \le b.
\end{align*}
Let $T$ be such a channel and consider the operators $\Pi_i=T^*(\ketbra{i}{i})$. These are positive and satisfy
\begin{equation}
\sum_i\Pi_i=T^*(\sum_i\ketbra{i}{i})=T^*(I)\le I,
\end{equation}
therefore $(\Pi_1,\ldots,\Pi_m,I-(\Pi_1+\cdots+\Pi_m))$ is a POVM. We estimate the probabilities as
\begin{equation}
\begin{aligned}
\Tr\rho_i\Pi_i & = \Tr\rho_i T^*(\ketbra{i}{i}) = \bra{i}T(\rho_i)\ket{i} \ge\bra{i}a_i\ket{i}\braket{i}{i}=a_i  \\
\Tr\sigma\Pi_i & = \Tr\sigma T^*(\ketbra{i}{i}) = \bra{i}T(\sigma)\ket{i} \le \bra{i}b\ket{i} = b_i.
\end{aligned}
\end{equation}
\end{proof}

Next we relate the asymptotic preorder to measurements on multiple copies of the same unknown state. This is not as straightforward as the extension in the case of a composite hypothesis, because the powers of a box of the form $(a_1\ketbra{1}{1},\ldots,a_m\ketbra{m}{m},B)$ are not of the same form. The 
\begin{lemma}\label{lem:stdboxpowers}
Let $a_1,\ldots,a_m,b_1,\ldots,b_m\in[0,1]$ and
\begin{equation}
B=\sum_{i=1}^m b_i\ketbra{i}{i}.
\end{equation}
as in Proposition~\ref{prop:multipleerror}.
Then for every $n\in\naturals$, $n\ge 1$ the inequalities
\begin{align}
(a_1\ketbra{1}{1},\ldots,a_m\ketbra{m}{m},B)^{\otimes n} & \preorderge(a_1^n\ketbra{1}{1},\ldots,a_m^n\ketbra{m}{m},B^n)
\intertext{and}
(a_1\ketbra{1}{1},\ldots,a_m\ketbra{m}{m},B)^{\otimes n} & \preorderle(a_1^n\ketbra{1}{1},\ldots,a_m^n\ketbra{m}{m},B^n)
\end{align}
hold.
\end{lemma}
\begin{proof}
For the first inequality, consider the map $T:\boundeds({\complexes^m}^{\otimes n})\to\boundeds(\complexes^m)$ given as
\begin{equation}
T(X)=\sum_{j=1}^m\ketbra{j}{jj\ldots j}X\ketbra{jj\ldots j}{j}.
\end{equation}
Then $T$ is completely positive and trace nonincreasing, $T((a_i\ketbra{i}{i})^{\otimes n})=a_i^n\ketbra{i}{i}$ and $T(B^{\otimes n})=B^n$.

For the second inequality, consider the map $T:\boundeds(\complexes^m)\to\boundeds((\complexes^m)^{\otimes n})$ defined as
\begin{equation}
T(X)=\sum_{j=1}^m\ketbra{jj\ldots j}{j}X{j}\ketbra{jj\ldots j}.
\end{equation}
Then $T$ is completely positive and trace nonincreasing, $T(a_i^n\ketbra{i}{i})=a_i^n\ketbra{ii\ldots i}{ii\ldots i}=(a_i\ketbra{i}{i})^{\otimes n}$ and
\begin{equation}
\begin{split}
T(B^n)
 & = \sum_{j=1}^m\ketbra{jj\ldots j}{j}\left(\sum_{i=1}^m b_i^n\ketbra{i}{i}\right)\ketbra{j}{jj\ldots j}  \\
 & = \sum_{i=1}^m b_i^n\ketbra{ii\ldots i}{ii\ldots i}  \\
 & \le B^{\otimes n}.
\end{split}
\end{equation}
\end{proof}

With the help of Lemma~\ref{lem:stdboxpowers} we can prove the analogue of Proposition~\ref{prop:compositeasymptotic}.
\begin{proposition}\label{prop:multipleasymptotic}
Let $(\rho_1,\ldots,\rho_m,\sigma)$ be a box and $r_1,\ldots,r_m,R_1,\ldots,R_m\ge 0$, and consider the operator $r=\sum_{i=1}^m r_i\ketbra{i}{i}$. The following are equivalent
\begin{enumerate}
\item there exists a sequence of POVMs $(\Pi_{n,1},\ldots,\Pi_{n,m},I-(\Pi_{n,1}+\cdots+\Pi_{n,m}))$ such that for all $i\in\{1,\ldots,m\}$ the observed probabilities behave as
\begin{equation}
\begin{aligned}
\Tr\rho_i^{\otimes n}\Pi_{n,i} & \ge 2^{-R_in+o(n)}  \\
\Tr\sigma^{\otimes n}\Pi_{n,i} & \le 2^{-r_in}.
\end{aligned}
\end{equation}
\item $(\rho_1,\ldots,\rho_m,\sigma)\asymptoticge(2^{-R_1}\ketbra{1}{1},\ldots,2^{-R_m}\ketbra{m}{m},2^{-r})$.
\end{enumerate}
\end{proposition}
\begin{proof}
The asymptotic relation is equivalent to
\begin{equation}
\begin{split}
(\rho_1^{\otimes n},\ldots,\rho_m^{\otimes n},\sigma^{\otimes n})
 & \preorderge u^{-k_n}(2^{-R_1}\ketbra{1}{1},\ldots,2^{-R_m}\ketbra{m}{m},2^{-r})^n  \\
 & \preorderge (2^{-R_1n-k_n}\ketbra{1}{1},\ldots,2^{-R_mn-k_n}\ketbra{m}{m},2^{-rn})
\end{split}
\end{equation}
by Lemma~\ref{lem:stdboxpowers}, where $k_n/n\to 0$. Proposition~\ref{prop:multipleerror} in turn says that this is equivalent to the existence of a sequence of POVMs $(\Pi_{n,1},\ldots,\Pi_{n,m},I-(\Pi_{n,1}+\cdots+\Pi_{n,m}))$ such that $\Tr\rho_i^{\otimes n}\Pi_{n,i}\ge 2^{-R_in-k_n}$ and $\Tr\sigma^{\otimes n}\Pi_{n,i}\le 2^{-r_in}$ for all $i\in\{1,\ldots,m\}$ with $k_n\in o(n)$.
\end{proof}

From Corollary~\ref{cor:asymptotic} we get that the exponents $R_1,\ldots,R_m,r_1,\ldots,r_m$ are achievable iff for all $i\in\{1,\ldots,m\}$ the inequality
\begin{equation}
R_i\ge\sup_{\alpha>1}\frac{\alpha-1}{\alpha}\left[r_i-\sandwicheddivergence{\alpha}{\rho_i}{\sigma}\right]
\end{equation}
holds.

\section*{Acknowledgement}

This work was supported by the \'UNKP-19-4 New National Excellence Program of the Ministry for Innovation and Technology and the J\'anos Bolyai Research Scholarship of the Hungarian Academy of Sciences. We acknowledge support from the Hungarian National Research, Development and Innovation Office (NKFIH) within the Quantum Technology National Excellence Program (Project Nr.~2017-1.2.1-NKP-2017-00001) and via the research grants K124152, KH129601.

\bibliography{refs}{}

\end{document}